\documentclass[nonacm=true,review=false]{acmart}
\usepackage{todonotes}

\AtBeginDocument{%
  \providecommand\BibTeX{{%
    \normalfont B\kern-0.5em{\scshape i\kern-0.25em b}\kern-0.8em\TeX}}}

\setcopyright{acmcopyright}
\copyrightyear{XXXX}
\acmYear{XXXX}
\acmDOI{XXXXXXX.XXXXXXX}

\acmJournal{TOCL}
\acmVolume{XX}
\acmNumber{X}
\acmArticle{111}
\acmMonth{8}

\usepackage{amsmath}
\usepackage{bussproofs}
\usepackage{xspace}
\usepackage{tikz}
\usetikzlibrary{arrows.meta}
\tikzset{>=Latex}    

\theoremstyle{acmplain}
\newtheorem{theorem}{Theorem}[section]

\newtheorem{proposition}[theorem]{Proposition}
\newtheorem{lemma}[theorem]{Lemma}
\newtheorem{corollary}[theorem]{Corollary}
\newtheorem{remark}[theorem]{Remark}
\theoremstyle{acmdefinition}
\newtheorem{definition}[theorem]{Definition}
\theoremstyle{remark}

\newcommand{\HoareTriple}[3]{\ensuremath{\{#1\}\,#2\,\{#3\}}}
\newcommand{\IncTriple}[3]{\ensuremath{[#1]\,#2\,[#3]}}
\newcommand{\SFHTriple}[3]{\ensuremath{\langle#1\rangle\,#2\,\langle#3\rangle}}
\newcommand{\SPCo}[1]{\ensuremath{P^s_{#1}}}
\newcommand{\SPC}[2]{\ensuremath{\SPCo{#1}(#2)}}
\newcommand{\ahl}{aHl\xspace}

\newcommand{\executes}[3]{\ensuremath{#2\stackrel{#1}{\longmapsto}#3}}

\newcommand{\limp}{\rightarrow}
\newcommand{\lpmi}{\leftarrow}
\newcommand{\Iff}{\Leftrightarrow}

\newcommand{\codecommand}[1]{\ensuremath{\mathtt{#1}}\xspace}
\newcommand{\ccskip}{\codecommand{skip}}
\newcommand{\ccif}{\codecommand{if}}
\newcommand{\ccthen}{\codecommand{then}}
\newcommand{\ccelse}{\codecommand{else}}

\newcommand{\ccwhile}{\codecommand{while}}
\newcommand{\ccdo}{\codecommand{do}}

\newcommand{\ccifte}[3]{\ccif\ #1\ \ccthen\ #2\ \ccelse\ #3}
\newcommand{\ccwhiledo}[2]{\ccwhile\ #1\ \ccdo\ #2}

\newcommand{\lh}{\codecommand{lh}}
\newcommand{\el}{\codecommand{el}}
\newcommand{\ttL}{\codecommand{L}}
\newcommand{\ttp}{\codecommand{p}}
\newcommand{\tti}{\codecommand{i}}
\newcommand{\ttacc}{\codecommand{acc}}
\newcommand{\ttfalse}{\codecommand{false}}
\newcommand{\tttrue}{\codecommand{true}}

\renewcommand{\labelenumi}{(\alph{enumi})}
\renewcommand{\theenumi}{(\alph{enumi})}

\newcommand{\sPK}{\codecommand{scriptPubKey}}
\newcommand{\sS}{\codecommand{scriptSig}}

\newcommand{\pKH}{\codecommand{P2PKH}}
\newcommand{\sPKpKH}{\ensuremath{\sPK_\pKH}\xspace}
\newcommand{\sSpKH}{\ensuremath{\sS_\pKH}\xspace}

\begin{document}

\title{Access Hoare Logic}

\author{Arnold Beckmann}
\authornote{Corresponding author}
\email{a.beckmann@swansea.ac.uk}
\orcid{0000-0001-7958-5790}
\author{Anton Setzer}
\email{a.g.setzer@swansea.ac.uk}
\orcid{0000-0001-5322-6060}
\affiliation{%
  \institution{Dept.~of Computer Science, Swansea University}
  \city{Swansea SA1 8EN}
  \state{Wales}
  \country{UK}
}


\begin{abstract}
  Following Hoare's seminal invention, now called \emph{Hoare logic},
  to reason about
  \emph{correctness of computer programs},
  we advocate a related but fundamentally different approach to reason about
  \emph{access security of computer programs} such as access control.
  We define the formalism, which we denote \emph{access Hoare logic},
  and present
  examples which demonstrate its usefulness
  and fundamental difference to Hoare logic.
  We prove soundness and completeness of access Hoare logic,
  and provide a link between access Hoare logic and standard Hoare logic.
  We also demonstrate a fundamental difference of access Hoare logic to other approaches, in particular incorrectness logic.
\end{abstract}



\maketitle
\section{Introduction}


The correctness of programs has been a
concern early on in the development of computers.
The seminal contribution by Tony Hoare~\cite{hoare:69}
allows to reason about programs instruction by instruction
using pre- and post-conditions.
It allows to ascertain that if the pre-condition is satisfied
before execution, and program execution terminates,
then the post-condition will be satisfied.

This form of ascertaining the correctness of programs has been and still is
very successful.
It is at the heart of several tools that are used in academia and industry
to specify and verify software systems, such as SPARK~\cite{SPARKAda:Homepage}
and Dafny~\cite{dafny:Homepage}.

However, there are security properties that do not naturally
fall within the framework of Hoare logic.
Access control is the property of restricting access to a resource.
It is a well-established area of computer security that includes
identification,
authorization, authentication, access approval, and
auditing~\cite{gollmann:2011},
where access approval is the process of deciding
whether to grant or reject an access request.
Access control has a broad range of applications, ranging from large
distributed systems to local systems such as electronic keys.
More recently, it gained importance in the context of blockchain and
cryptocurrencies:
Smart contracts are often written to control access
to digital assets, such as cryptocurrencies, tokens,
or other forms of digital assets managed by blockchain systems,
thereby combining several properties of general access control.
We will use the term \emph{access security} to speak about such properties in general.

For a program
to be access secure,
we need to guarantee that the final state
(e.g., access to a hotel room via electronic keys)
will only be granted
if a given set of conditions has been met before execution
(e.g., possession of valid access codes).
That is, for access
security
we are interested in the \emph{reverse} direction
of Hoare logic, which we will denote \emph{access Hoare logic (\ahl)},
where we will
reason
from post-conditions to pre-conditions.
While normal Hoare logic expresses a pre-condition to be \emph{sufficient} for
the post-condition to hold after executing a program,
access Hoare logic requires the pre-condition to be \emph{necessary}.
More precisely, access Hoare logic expresses that
if the program starting in state $s$ terminates in a state $s'$ which
fulfills the post-condition,
then it is necessary that state $s$ fulfills the pre-condition.
To make the difference between access Hoare logic and standard Hoare logic more apparent, we observe that
while \emph{False} is always a valid pre-condition in normal Hoare logic,
in case of access Hoare logic this is given by \emph{True}.

In this note, we define
access Hoare triples and access Hoare logic, and use it
to describe
access security for examples.
We show that access Hoare logic is sound and complete for access Hoare triples 
under a general, semantical interpretation.
We provide a link between Hoare logic and access Hoare logic in that
weakest pre-conditions for usual Hoare logic correspond to strongest
pre-conditions for access Hoare logic on states where the program
terminates (see Section~\ref{sec:precondition}).
We also demonstrate a fundamental difference of access Hoare logic to incorrectness logic \cite{OHearn:IncorrectnessLogic} and related approaches: 
while incorrectness logic provides a complementary view to standard Hoare logic in the context of forward transformers,
access Hoare logic plays a similar role in the context of backwards transformers (see Section~\ref{sec:comparisons}).

Alhabadi et al.~\cite{alhabadi:etal:2022} have recently discovered
that weakest pre-conditions for Hoare logic 
can be used to express access security for Bitcoin scripts.
This paper expands on those insights and provides a reason why that
is the case.

\section{Examples}\label{sec:examples}

We demonstrate our point of view by considering three examples: 
The first involves electronic keys, 
the second owning and transferring Bitcoins,
and the third granting access depending on the presence of a key in a
list of keys.
We will revisit the examples in Section~\ref{sec:examplesFormal} to formally analyze their access security after introducing our framework.

\subsection{Access Security for Electronic Keys}\label{ex:HotelKeys}
\newcommand{\deviceKey}{\codecommand{dk}}
\newcommand{\cardKeyOne}{\codecommand{ck1}}
\newcommand{\cardKeyTwo}{\codecommand{ck2}}
\newcommand{\access}{\codecommand{acc}}
\newcommand{\cctrue}{\codecommand{true}}
\newcommand{\ccfalse}{\codecommand{false}}
\newcommand{\keyCurrentGuest}{\codecommand{keyCurrentGuest}}

Our first example is a digital key system in a hotel that uses cards to open doors,
which was specified and verified  by
Jackson~\cite{jackson2012:softwareAbstraction:2ndEd}
using \emph{Alloy},
and by Nipkow~\cite{Nipkow-VerficationHotelKeySystem} using
\emph{Isabelle/HOL}.
In this system, the battery-powered door locks are not connected to a
network.
Instead, their state is controlled by cards that are inserted
into them, or, in wireless systems, moved close to them.
A lock holds an internal electronic key \deviceKey
which controls its state.
A card holds two electronic keys,
\cardKeyOne for the first key and \cardKeyTwo for the second.
When a new guest receives a card (for example, from the hotel reception),
it will contain the key corresponding to a previous guest as
\cardKeyOne, and a new fresh key as \cardKeyTwo.
An access control program for the door lock suitable for use in such an
electronic hotel key system
should give access to the room in an access-secure way, by setting
a variable \access to true:
    When the guest uses their card for the first time,
access will be granted if the previous guest's
key is stored as the door key
($\deviceKey{} = \cardKeyOne{}$),
in which case the door key is updated to the new key \cardKeyTwo{},
preventing the previous guest from gaining access to the room
thereafter.
The next time the new guest swipes their card,
access will be granted because the second key on their card,
\cardKeyTwo{}, now matches the door key ($\deviceKey{} =  \cardKeyTwo{}$).

For this example, a pre-condition for verifying access security is given by
$\cardKeyOne = \deviceKey \lor \cardKeyTwo = \deviceKey$, expressing that
either of the first key or the second key is the door key. A post-condition is
$\access = \cctrue$,
expressing that access is granted.

A program to solve this task may look as follows:
\begin{verbatim}
P0:  if not (dk == ck1)
       then acc := (dk == ck2)
       else dk := ck2; acc := true
\end{verbatim}
However, there is an ambiguity in the way this code is written, in that it
can be
parsed in different ways.
The intended way is that both instructions after \verb|else| would be
executed in the else part, i.e. the program would be parsed as
\begin{verbatim}
P1:  if not (dk == ck1)
       then {acc := (dk == ck2)}
       else {dk := ck2; acc := true}
\end{verbatim}
Another reading would give higher priority to '\verb|;|', leading to
\begin{verbatim}
P2:  {if not (dk == ck1)
       then {acc := (dk == ck2)}
       else {dk := ck2}};
     acc := true
\end{verbatim}

Both \verb|P1| and \verb|P2| are correct w.r.t.~%
the given pre- and post-condition in standard Hoare logic.
However, only \verb|P1| is access secure w.r.t.~the
given pre- and post-condition,
while \verb|P2| is not:
the given pre-condition is not necessary for the post-condition
to hold after executing \verb|P2|,
as \verb|P2| sets \verb|acc| to \verb|true| in any case.
See Section~\ref{sec:exampleHKformal} for further details.

\subsection{Access Security for Bitcoin}\label{ex:Bitcoin}

%

The cryptocurrency Bitcoin is managed by a blockchain that stores
transactions of Bitcoins on a ledger.
The transfer of Bitcoins from a current owner, Bob, 
to a recipient, 
Alice, 
via transactions
is controlled by small programs called scripts:
Bob's Bitcoins 
are protected by a \emph{locking script}, while the new transaction that
tries to transfer the Bitcoins to 
Alice
contains an \emph{unlocking script}.
To transfer the Bitcoins to 
Alice,
the unlocking script, followed
by the locking script, are executed;
if this execution succeeds (by producing the output \emph{true}),
the transaction is successful and the Bitcoins are transferred to 
Alice;
otherwise, the transaction fails and the Bitcoins will not be
transferred.

In this context, the focus is on the access security of the locking
script:
In any successful run of the unlocking and then the locking script,
the unlocking script is required to finish in a state which fulfills
a suitable pre-condition, see Fig.~\ref{fig-executing-lock-unlock-scripts}.
For a typical locking script
like the locking script $\sPK_\pKH$
of the standard pay-to-public-key-hash (P2PKH) script
\cite{bitcoin:wiki},
a typical pre-condition
would express that an address and a signature are
provided, where the address hashes to a value stored in the locking
script, and the signature matches the address.
Further details are given in Section~\ref{sec:exampleBitcoinformal}.

\begin{figure*}[ht]
\centering
\begin{tikzpicture}
  \node[align=center] at (3,2.5) {direction of execution};
  \draw[->] (0,2) -- (6,2);
  \node[draw, minimum width=4cm, minimum height=10mm, align=center]
  (box1) at (0,1) {unlocking script};
  \node[draw, minimum width=4cm, minimum height=10mm, align=center]
  (box2) at (6,1) {locking script};
  \draw[-] (1.5,0) -- (3,0.7);
  \draw[-] (4.5,0) -- (3,0.7);
  \node[align=center] at (3,0) {pre-condition};
  \draw[-] (7.5,0) -- (9,0.7);
  \draw[-] (10.5,0) -- (9,0.7);
  \node[align=center] at (9,0) {post-condition};
\end{tikzpicture}
  \caption{Executing the combined unlocking and locking script}\label{fig-executing-lock-unlock-scripts}
\end{figure*}

\subsection{Access Security for While Loops}\label{ex:CheckList}

We consider a program to
grant access by setting a variable \verb|acc| to true
if a given passkey \verb|p| is amongst a list of stored passkeys
\verb|L|,
where \verb|lh(L)| denotes the length of list \verb|L|,
and \verb|el(i,L)| denotes the \verb|i|-th element of list~\verb|L|.
The output of this program (i.e.~the variable $\ttacc$) may then be
used by another system to e.g.~open a door.

\begin{verbatim}
  inputs:  p : Keys
           L : List(Keys)
  outputs: acc : Boolean

  program CheckList
    i := 1;
    acc := false;
    while i<=lh(L) do
      if el(i,L) == p then acc := true endif;
      i := i+1;
    endwhile
\end{verbatim}

Suitable pre- and post-conditions for certifying the access security of
this program should reflect the requirement that access should only be
granted if a given passkey $\ttp$ is included in a list of stored passkeys~$\ttL$.
A suitable post-condition is given by $\ttacc=\tttrue$.
A suitable pre-condition for when to give access
is that $\ttp$ occurs in $\ttL$, which may be expressed as
$\exists j{\le}\lh(\ttL) ( \el(j,\ttL) = \ttp)$.
A suitable invariant for the while loop that reflects the intended
behaviour of the program may be
$\ttacc=\tttrue\ \lor\
\exists j{\le}\lh(\ttL) ( j\ge i \land \el(j,\ttL) = \ttp)$
where $i$ denotes the current index of the while loop.
The above loop invariant reflects the intention that within the while
loop, to reach the post-condition it is necessary that either the
post-condition has been reached already ($\ttacc=\tttrue$),
or that a matching passkey in the list $\ttL$ will be found in the
remaining steps of the loop
($\exists j{\le}\lh(\ttL) ( j\ge i \land \el(j,\ttL) = \ttp)$).
A detailed account of the access security of \verb|CheckList| is given in Section~\ref{sec:exampleCLformal}.



\section{Access Hoare Triples}\label{sec:aH triples}

Following Hoare \cite{hoare:69}, we take assertions to be semantical, that is, 
with no fixed syntax, and no restrictions on the expressions in the language.
The reason for this approach is that we want to focus on the general properties of access Hoare logic and their comparison to Hoare logic and related frameworks.

Hoare logic is using Hoare triples to describe
how the execution of code impacts on the state of the machine.
The Hoare triple \HoareTriple{P}{C}{Q} consists of assertions $P$
and $Q$, and a piece of code $C$.
$P$~is called the \emph{pre-condition}, $Q$ the \emph{post-condition}.
\HoareTriple{P}{C}{Q} expresses that when $P$ is true and $C$ is executed and
terminates
then $Q$ is true.
In other words, $P$ is a \emph{sufficient} condition that $Q$ holds after
successfully executing $C$.
It can be expressed more formally by quantifying over all possible states:
\begin{equation}\label{formal H triple}
  \forall s,s' [ \executes{C}{s}{s'} \land P(s) \limp Q(s')]
\end{equation}
where a state is a mapping of variables to values,
and $\executes{C}{s}{s'}$ expresses that the execution of $C$ takes
the executing machine from state $s$ to state~$s'$.

We are interested in the opposite direction, that $P$ is a \emph{necessary}
condition for $Q$
to hold after executing $C$.
We again use a triple \SFHTriple{P}{C}{Q}, which we denote \emph{access Hoare
triple}, consisting again of a pre-condition $P$, a post-condition $Q$, and a
program $C$.
However, now \SFHTriple{P}{C}{Q} expresses that when the execution of $C$
results in $Q$ being true, then
$P$ must have been true before, which can be expressed more formally as
\begin{equation}\label{formal aH triple}
  \forall s,s' [ \executes{C}{s}{s'} \wedge Q(s') \limp P(s)]
\end{equation}
In this case, we also say that \emph{\SFHTriple{P}{C}{Q} is true} or \emph{valid}.

There is a close relationship between access Hoare triples and Hoare triples
by introducing negation:

\begin{theorem}\label{equivalence aHt and Ht}
    $\SFHTriple{P}{C}{Q}$ is equivalent to
    $\HoareTriple{\neg P}{C}{\neg Q}$
\end{theorem}

\begin{proof}
    This follows since $Q(s') \limp P(s)$ is
    equivalent to $\neg P(s) \limp \neg Q(s')$.
\end{proof}


At this point we could introduce a calculus for access Hoare logic
indirectly as negated Hoare logic in the sense of
Theorem~\ref{equivalence aHt and Ht}.
However, it is of advantage to have a direct calculus for access
Hoare logic.
For example, a calculus for access Hoare logic would form the basis for systems that can verify
access security in the sense of access Hoare logic using tools like
theorem provers, similar to SPARK~\cite{SPARKAda:Homepage}
and Dafny~\cite{dafny:Homepage} for verifying correctness based on standard
Hoare logic.
%
%
%
A reduction from access Hoare logic to standard Hoare logic,
as outlined above, would introduce the contraposition of implications
in verification proofs,
i.e.\ $\neg B\limp\neg A$ instead of $A\limp B$.%
\footnote{In more detail,
  a proof of an access Hoare triple \SFHTriple PCQ
  will usually involve subproofs of access Hoare triples
  of  the form \SFHTriple{P'}{C'}{Q'} and \SFHTriple{P''}{C''}{Q''},
  for certain subprograms $C'$ and $C''$,
  where $C'$ is immediately followed by $C''$.
  In such a situation, arguing access Hoare logic style
  from post- to pre-conditions
  requires proving the implication $P''\limp Q'$
  from the pre-condition $P''$ of the second subproof
  to the post-condition $Q'$ of the first.
  If we consider this situation by transformation to Hoare logic,
  we are faced with \HoareTriple{\neg P}{C}{\neg Q} and subproofs of
  \HoareTriple{\neg P'}{C'}{\neg Q'} and \HoareTriple{\neg P''}{C''}{\neg Q''},
  now having to prove the implication $\neg Q'\limp\neg P''$
  while arguing Hoare logic style from pre- to post-conditions.
  }
For example, a straightforward implication such as
$x=2\limp x+2=4$
may translate into its contrapositive
$\neg(x+2=4)\limp\neg(x=2)$.
While logically equivalent,\footnote{In classical logic} the latter is
more difficult to understand and deal with in verification tasks.
In general, this may even break
verifiability.\footnote{SPARK~\cite{SPARKAda:Homepage},
  which supports  Hoare logic,
  tries to prove formulas generated by a verification process using an SMT solver.
  If that fails, SPARK offers to use the
  interactive theorem prover Coq/Rocq \cite{sparkAda:Doc:Coq} instead.
  Another option when facing unprovability in SPARK, is to analyze formulas manually
  in Why3 \cite{why3:homepage,adacore:SPARKdoc:ArchitectureQualityAssurance:Why3},
  in order to find a reason for unprovability.
  In the last two cases, the syntactical form of formulas matters:
  For example,
  Coq/Rocq is based on intuitionistic logic, hence transformation of
  formulas need to respect intuitionistic logic.
}%

We will introduce
a direct calculus for access Hoare logic
by stating its rules similar to those for Hoare logic.
Soundness and completeness of such a calculus could be proven
using soundness and completeness of Hoare logic and
Theorem~\ref{equivalence aHt and Ht}.
As it turns out, proving soundness and completeness directly
does not take much longer, and has the added benefit
of motivating the rules, and keeping the paper selfcontained.
Furthermore, it also supports a more direct and intuitive
understanding of the rules of access Hoare logic as a formal method
for access security.

\begin{remark}\label{equivalenceIntuitionistic}
    Following the previous discussion, we analyze
    Theorem~\ref{equivalence aHt and Ht} from the perspective of
    intuitionistic logic.
    The direction
    $\SFHTriple{P}{C}{Q} \rightarrow \HoareTriple{\neg P}{C}{\neg Q}$
    is valid in intuitionistic logic,
    while for the reverse direction we only obtain
    $\HoareTriple{\neg P}{C}{\neg Q} \to \SFHTriple{\neg \neg P}{C}{\neg\neg Q}$.
    Observe that $\neg\neg P\limp P$ is valid in classical logic, but not in intuitionistic logic in general.
\end{remark}

\section{The Calculus}\label{sec:calculus}

We define the calculus for \SFHTriple{P}{C}{Q},
which we denote \emph{access Hoare logic},
following
that for Hoare logic~\cite{hoare:69}, and show that it
is sound with respect to the interpretation of access Hoare triples as given
in~\eqref{formal aH triple} in the previous section.
The development will focus on the \ccwhile language similar to
\cite{hoare:69}, a standard imperative language.
We will show completeness in Section~\ref{sec:completeness}.

\subsection{Axiom for empty statement}

The empty statement \ccskip does not change the state of the program,
thus whatever holds after \ccskip must have held before.
\begin{prooftree}
    \AxiomC{}
    \UnaryInfC{\SFHTriple{P}{\ccskip}{P}}
\end{prooftree}
Soundness of this rule is immediate.

\subsection{Axiom scheme of assignment}
An assignment of the from
\[
  V:=E
\]
with $V$ a variable and $E$ an expression without side effects, but
possibly containing $V$, updates variable $V$ with the result of evaluating
$E$.
Any assertion $P$ that is true after the assignment is made must also have
been true of the value resulting from the evaluation of $E$.
\begin{prooftree}
    \AxiomC{}
    \UnaryInfC{\SFHTriple{P[E/V]}{V:=E}{P}}
\end{prooftree}
Here
$P[E/V]$ denotes the result of replacing any free occurrence of $V$ in $P$ by
$E$.


Soundness follows from the same considerations as for
standard Hoare logic:
Let $s,s'$ be such that $\executes{V:=E}{s}{s'}$, then
$P[E/V](s)$ is equivalent to $P(s')$.

\subsection{Consequence rule}
The consequence rule is the main deviation from Hoare's original calculus,
reflecting that we are interested in expressing
that the pre-condition is
necessary for the post-condition.
Thus, contrary to standard Hoare logic, we allow to weaken the pre-condition,
and to strengthen the post-condition.
\begin{prooftree}
    \AxiomC{$P_1\lpmi P_2$}
    \AxiomC{\SFHTriple{P_2}{S}{Q_2}}
    \AxiomC{$Q_2\lpmi Q_1$}
    \TrinaryInfC{\SFHTriple{P_1}{S}{Q_1}}
\end{prooftree}

For soundness,
let $s,s'$ be such that $\executes{S}{s}{s'}$ and $Q_1(s')$ hold.
We need to show that $P_1(s)$ is true, assuming 
validity of the premises of
this rule.
From $Q_1(s')$ we obtain $Q_2(s')$ from the right premise.
Validity of the middle premise then shows $P_2(s)$, from which the 
validity of
the left premise implies~$P_1(s)$.

\subsection{Rule of composition}

As in standard Hoare logic, $S;T$ denotes the sequential composition of $S$
and $T$,
where $S$ executes prior to $T$.
\begin{prooftree}
    \AxiomC{\SFHTriple{P}{S}{R}}
    \AxiomC{\SFHTriple{R}{T}{Q}}
    \BinaryInfC{\SFHTriple{P}{S;T}{Q}}
\end{prooftree}

For soundness,
let $s,s'$ be such that $\executes{S;T}{s}{s'}$ and $Q(s')$ hold.
We need to show that $P(s)$ is true, assuming 
validity of the premises of
this rule.
Let $s''$ such that $\executes{S}{s}{s''}$ and $\executes{T}{s''}{s'}$.
Validity of the right premise yields $R(s'')$, from which
validity of the
left shows $P(s)$.

\subsection{Conditional rule}

The conditional rule also differs from standard Hoare logic.
The assertion $B$ is used to `choose a branch' corresponding to the `then' or
`else' part
of \ccifte{B}{S}{T}.
For example, the left premise of the rule, in both standard and access Hoare
logic,
expresses `if $B$ then $(P) S (Q)$':
For standard Hoare logic, this expands to
'if $B$ then if $P$ and $S$ execute successfully, then $Q$',
which is logically equivalent to
'if $B\land P$ and $S$ execute successfully, then $Q$'.
For access Hoare logic, the expression expands to
'if $B$ then if $Q$ and $S$ execute successfully, then $P$',
which is logically equivalent to
'if $Q$ and $S$ executes successfully, then $B\limp P$'.
Using a similar consideration for the right premise,
the conditional rule takes the following form:

\begin{prooftree}
    \AxiomC{\SFHTriple{B\limp P}{S}{Q}}
    \AxiomC{\SFHTriple{\neg B\limp P}{T}{Q}}
    \BinaryInfC{\SFHTriple{P}{\ccifte{B}{S}{T}}{Q}}
\end{prooftree}
To formally argue for soundness of this rule,
let $s,s'$ be such that $\executes{\ccifte{B}{S}{T}}{s}{s'}$ and $Q(s')$ hold.
We need to show that $P(s)$ holds, assuming 
validity of the two premises.

If $B(s)$ is true, then $\executes{S}{s}{s'}$, so $B(s)\limp P(s)$ is true
by 
validity of the left premise.
Thus, $P(s)$ holds.
Otherwise, $B(s)$ is false, and $\executes{T}{s}{s'}$, hence $\neg B(s)\limp
P(s)$ is true by 
validity of the right premise.
Again, $P(s)$ holds.

\subsection{While rule}

Similarly to the conditional rule, the pre-condition of the premise changes from
conjunction to implication compared to standard Hoare logic.
Furthermore, we weaken the post-condition
to only condition reachable states,
namely on those satisfying
$\neg B$
-- this is needed to be able to prove completeness in
Section~\ref{sec:completeness}.

\begin{prooftree}
    \AxiomC{\SFHTriple{B\limp P}{S}{P}}
    \UnaryInfC{\SFHTriple{P}{\ccwhiledo{B}{S}}{\neg B\limp P}}
\end{prooftree}
To prove soundness of this rule,
let $s,s'$ be such that $\executes{\ccwhiledo{B}{S}}{s}{s'}$ and
$\neg B(s') \limp P(s')$ hold.
We need to show that $P(s)$ holds, assuming 
validity of the premise.

As \ccwhiledo{B}{S} terminates, there are some $k$ and states $s=s_0$, \dots,
$s_k=s'$ such that
\[
    \forall i<k [ B(s_i)\land \executes{S}{s_i}{s_{i+1}}]
\]
Furthermore, termination implies $\neg B(s_k)$, hence $P(s_k)$ by assumption.
Hence, \executes{S}{s_{k-1}}{s_{k}} and 
validity of the premise
\SFHTriple{B\limp P}{S}{P} show
$B(s_{k-1})\limp P(s_{k-1})$.
As we also have $B(s_{k-1})$, we obtain $P(s_{k-1})$ by modus ponens.
Inductively,
we obtain $P(s_0)$
as required.


\section{Weakest and Strongest Pre-Conditions}\label{sec:precondition}

The aim of this section is to show a connection
between Hoare triples and access Hoare triples,
namely that modulo termination,
weakest pre-conditions for Hoare triples,
and strongest pre-conditions for access Hoare triples
coincide.
We start by repeating standard definitions of
weakest pre-conditions for Hoare triples,
and provide a definition of strongest pre-conditions for access Hoare triples.
We then give explicit characterizations of weakest pre-conditions for Hoare
triples
and strongest pre-conditions for access Hoare triples.
We conclude by showing that, for a fixed program and post condition,
the strongest pre-condition for Hoare triples equals to the
intersection of the weakest pre-condition for Hoare Triples with the set of states
on which the program terminates.


\begin{definition}[Weakest pre-condition for Hoare triples]
Given a program $C$ and a post-condition $Q$, a \emph{weakest pre-condition} is
a predicate $P'$ such that
\[
 \forall P,\, \HoareTriple{P}{C}{Q}\ \Leftrightarrow\ P\limp P'
\]
\end{definition}

\begin{definition}[Strongest pre-condition for access Hoare triples]
Given a program $C$ and a post-condition $Q$, a \emph{strongest pre-condition}
is a predicate $P'$ such that
\[
    \forall P,\, \SFHTriple{P}{C}{Q}\ \Leftrightarrow\ P'\limp P
\]
\end{definition}

It follows immediately from the definitions that weakest and strongest
pre-conditions are
pre-conditions for Hoare logic and access Hoare logic, respectively,
and therefore are unique.




\begin{lemma}
\label{characterizationPw}
Given a program $C$ and a post-condition $Q$, let $P^w$ be defined by
\[
   P^w(s) \ \Iff\ \forall s',\, \executes{C}{s}{s'}\limp Q(s')
\]
Then $P^w$ is the weakest pre-condition for $C$ and $Q$ w.r.t.~Hoare triples.
\end{lemma}


\begin{proof}
To show that $P^w$ is the weakest pre-condition, it we need to show that
\[
  \forall P,\, \HoareTriple{P}{C}{Q} \Iff P\limp P^w
\]
Let $P$ be given.
For the direction from right to left, assume $P\limp P^w$.
Let $s,s'$ be given with \executes{C}{s}{s'} and $P(s)$,
hence $P^w(s)$.
By the definition of $P^w$ and the choice of $s'$, we have $Q(s')$.
This shows $\HoareTriple{P}{C}{Q}$.

For the other direction, assume $\HoareTriple{P}{C}{Q}$ and let $s$ with
$P(s)$ be given.
We need to show $P^w(s)$.
To this end, assume $s'$ such that $\executes{C}{s}{s'}$.
By $\HoareTriple{P}{C}{Q}$ we obtain $Q(s')$, as required.
Together, this shows $P\limp P^w$.
\end{proof}

\begin{lemma}
\label{characterizationPs}
Given a program $C$ and a post-condition $Q$, let $P^s$ be defined by
\[
   P^s(s) \ \Iff\ \exists s',\, \executes{C}{s}{s'}\land Q(s')
\]
Then $P^s$ is the strongest pre-condition for $C$ and $Q$ w.r.t.~access Hoare
triples.
\end{lemma}

\begin{proof}
To show that $P^s$ is the strongest pre-condition, we need to show that
\[
  \forall P,\, \SFHTriple{P}{C}{Q} \Iff P^s\limp P
\]
Let $P$ be given.
For the direction from right to left, assume $P^s\limp P$.
Let $s,s'$ be given with \executes{C}{s}{s'} and $Q(s')$.
Then $P^s(s)$ by the definition of $P^s$.
Hence $P(s)$ by assumption.
This shows $\SFHTriple{P}{C}{Q}$.

For the other direction, assume $\SFHTriple{P}{C}{Q}$ and let $s$ with $P^s(s)$
be given.
By the definition of $P^s$, there exists $s'$ such that
$\executes{C}{s}{s'}$ and $Q(s')$.
By $\SFHTriple{P}{C}{Q}$ we obtain $P(s)$.
This shows $P^s\limp P$.
\end{proof}

\begin{corollary}
\label{corollaryPwEqualPsU}
Let $P^w$ be the weakest pre-condition for $C$ and $Q$ w.r.t.~Hoare triples, and
$P^s$ be the strongest pre-condition for $C$ and $Q$ w.r.t.~access Hoare triples.
Then $P^s$ is equal to the intersection of $P^w$ with the set of states in which $C$ terminates.
%
\end{corollary}


\begin{proof} Let $T$ be the set of states in which $C$ terminates.
Then we can argue
\[\begin{array}[t]{@{}lcl}
P^s(s) &\Iff& \exists s',\, \executes{C}{s}{s'}\land Q(s')\\
       &\Iff& (\exists s',\, \executes{C}{s}{s'})\land (\forall s',\, \executes{C}{s}{s'}\to Q(s'))\\
&\Iff& T(s) \land P^w(s)
\end{array} \]
where we have used that $C$ is deterministic.
\end{proof}

\begin{corollary}
\label{corollaryPwEqualPsIfTerminating}
  For programs that terminate on all inputs,
  weakest pre-conditions for Hoare
  logic and strongest pre-conditions for access Hoare logic coincide.
\end{corollary}

\begin{remark}\label{remGeneralisationintuit}
\begin{enumerate}
\item
\label{remGeneralisation}
    The assertions of Lemmas \ref{characterizationPw}, \ref{characterizationPs}, and Corollaries
    \ref{corollaryPwEqualPsU}, \ref{corollaryPwEqualPsIfTerminating}
    hold in a general, abstract setting as long as the execution relation is deterministic, that is, they hold
    for any set $S$ (representing states) and any relation $\executes{C}{s}{s'}$ that satisfies
    \[ \forall s,s',s'',\ \executes{C}{s}{s'} \land \executes{C}{s}{s''} \to s' = s'' \enspace.\]

\item
    Corollary \ref{corollaryPwEqualPsU} characterizes $P^s$.
    We can also characterize $P^w$, using a similar proof:
    $P^w$ consists of all states $s$ which satisfy
    that if $C$ terminates on $s$,
    then $s$ is in $P^s$.

\end{enumerate}
\end{remark}

\section{Soundness and Completeness}\label{sec:completeness}

Access Hoare logic, denoted \ahl{}, is given by
access Hoare triples
and their semantics as defined in Section~\ref{sec:aH triples},
and by the set of rules defined in Section~\ref{sec:calculus}.
When introducing the rules, we showed already that all axioms and rules
of \ahl{} are sound.
%
%
%
%
%
This implies the soundness of provable access Hoare triples:
\begin{theorem}[Soundness of \ahl]\label{thm:soundness}
If $\SFHTriple{P}{C}{Q}$ is provable in \ahl, then $\SFHTriple{P}{C}{Q}$ is
true.
\end{theorem}

We will now show the converse, that \ahl is also complete:
\begin{theorem}[Completeness of \ahl]\label{thm:completeness}
If $\SFHTriple{P}{C}{Q}$ is true, then $\SFHTriple{P}{C}{Q}$ is provable in
\ahl.
\end{theorem}

We could prove this theorem by relying on the completeness of Hoare logic and
Theorem \ref{equivalence aHt and Ht}. Since, as discussed before,
we think there is value in treating access Hoare logic directly, we
give here a direct proof.


To this end, we prove a slightly stronger proposition.
In the previous section we have defined the strongest pre-conditionfor program $C$ and
post-condition~$Q$, here we define $\SPC{C}{Q}$ as an operator on $Q$ that is indexed by $C$:
\[ \SPC CQ \ :=\ \{s\colon \exists s',\, \executes{C}{s}{s'} \wedge Q(s') \}   \]

\begin{proposition}\label{prop:completeness}
  For any program $C$ and predicate $Q$,
  $\SFHTriple{\SPC CQ}{C}{Q}$ is provable in \ahl.
\end{proposition}

This proposition implies completeness:
Assume $\SFHTriple{P}{C}{Q}$ is true.
As $\SPC CQ$ is the strongest pre-condition for $C$ and $Q$, we obtain
\[ \SPC CQ\subseteq P  \]
The previous proposition shows that $\SFHTriple{\SPC CQ}{C}{Q}$ is provable in \ahl.
Thus, using the consequence rule, we obtain that $\SFHTriple{P}{C}{Q}$ is
provable in \ahl.

\begin{proof}[Proof of Proposition~\ref{prop:completeness}]
    We prove the assertion by structural induction on~$C$.
    For the empty statement \ccskip we compute that $\SPCo{\ccskip}$ is the
    identity function on predicates.
    Hence, the assertion follows immediately from the axiom for the empty
    statement.

    In the case of assignment $V:=E$, we compute
    \begin{align*}
      s\in\SPC{V:=E}{Q}\ \  
      &\Leftrightarrow\ \ \exists s',\, \executes{V:=E}{s}{s'} \wedge Q(s') \\
      &\Leftrightarrow\ \  s\in Q[E/V]
    \end{align*}
    Hence, $\SPC{V:=E}Q \ =\ Q[E/V]$ and the assertion again follows immediately from the corresponding axiom.

    In case of sequential composition $S;T$, we can form the following \ahl
    derivation by applying the rule of composition to the induction
    hypotheses for $S$ and $T$:
    \begin{prooftree}
    \AxiomC{i.h.}
    \noLine
    \UnaryInfC{\SFHTriple{\SPC S{\SPC TQ}}{S}{\SPC TQ}}
    \AxiomC{i.h.}
    \noLine
    \UnaryInfC{\SFHTriple{\SPC TQ}{T}{Q}}
    \BinaryInfC{\SFHTriple{\SPC S{\SPC TQ}}{S;T}{Q}}
    \end{prooftree}
    An easy calculation shows
    $\SPCo{S;T} = \SPCo{S} \circ \SPCo{T}$,
    hence the assertion follows.
\renewcommand{\labelenumi}{(\arabic{enumi})}
\renewcommand{\theenumi}{(\arabic{enumi})}

    In the case of the conditional statement $C = \ccifte{B}{S}{T}$,
    let $P=\SPC CQ$.  We claim
    \begin{align}
    \label{eq1}
      \SPC SQ &\limp (B\limp P) \\
    \label{eq2}
      \SPC TQ &\limp (\neg B\limp P)
    \end{align}

    For \eqref{eq1}, assume $s\in\SPC SQ$ and $B(s)$.
    By definition, there exists $s'$ such that $\executes{S}{s}{s'} \wedge
    Q(s')$.
    As $B(s)$ holds, we obtain $\executes{C}{s}{s'} \wedge Q(s')$, thus
    $s\in\SPC CQ=P$.
    Similar for \eqref{eq2}. 

    Using \eqref{eq1} and \eqref{eq2} and the induction hypotheses for $S$
    and $T$,
    we can form the following derivation:
    \begin{prooftree}
    \AxiomC{\eqref{eq1}}
    \noLine
    \UnaryInfC{$\SPC SQ \limp (B\limp P)$}
    \AxiomC{i.h.}
    \noLine
    \UnaryInfC{\SFHTriple{\SPC SQ}{S}{Q}}
    \BinaryInfC{\SFHTriple{B\limp P}{S}{Q}}
    \AxiomC{\eqref{eq2}}
    \noLine
    \UnaryInfC{$\SPC TQ \limp (\neg B\limp P)$}
    \AxiomC{i.h.}
    \noLine
    \UnaryInfC{\SFHTriple{\SPC TQ}{T}{Q}}
    \BinaryInfC{\SFHTriple{\neg B\limp P}{T}{Q}}
    \BinaryInfC{\SFHTriple{P}{\ccifte{B}{S}{T}}{Q}}
    \end{prooftree}

    For the while statement $C = \ccwhiledo{B}{S}$, let $P=\SPC CQ$.
    We claim  
    \begin{equation}\label{eq while}
        \SPC SP\limp (B\limp P)    
    \end{equation}
    
    To prove \eqref{eq while}, assume that $s\in\SPC SP$ and $B(s)$ hold.
    We need to show that $P(s)$ holds.
    From $s\in\SPC SP$ we obtain 
    $\executes{S}{s}{s'}$ and $P(s')$ for some $s'$.   
    $P(s')$ implies $\executes{C}{s'}{s''}$ and $Q(s'')$
    for some~$s''$.
    The first part, $\executes{C}{s'}{s''}$, implies that
    there are $s_1,\dots,s_k$ such that $s_1=s'$, $s_k=s''$, and
    \begin{equation}\label{eq3}
        \executes{S}{s_i}{s_{i+1}}\wedge B(s_i)
    \end{equation}
    for $i=1,\dots,k{-}1$.
    Let $s_0=s$, then \eqref{eq3} also holds for $i=0$.
    Hence $\executes{C}{s}{s''}$, which, together with $Q(s'')$, shows $P(s)$.
    Thus, using the induction hypothesis, we can derive
    \begin{prooftree}
    \AxiomC{\eqref{eq while}}
    \noLine
    \UnaryInfC{$\SPC SP\limp (B\limp P)$}
    \AxiomC{i.h.}
    \noLine
    \UnaryInfC{\SFHTriple{\SPC SP}{S}{P}}
    \BinaryInfC{\SFHTriple{B\limp P}{S}{P}}
    \UnaryInfC{\SFHTriple{P}{\ccwhiledo{B}{S}}{\neg B\limp P}}
    \AxiomC{$\star$}
    \noLine
    \UnaryInfC{$Q\limp (\neg B\limp P)$}
    \BinaryInfC{\SFHTriple{P}{\ccwhiledo{B}{S}}{Q}}
    \end{prooftree}
    The final step is to show $\star$: Assuming $Q(s')$ and $\neg B(s')$, we
    need to show that $P(s')$ holds.
    We observe that $\neg B(s')$ implies $\executes{C}{s'}{s'}$, which
    together with $Q(s')$
    shows~$P(s')$.
\end{proof}

We finish this section by drawing two interesting consequences
that the following rules for conjunction and disjunction are admissible in
access Hoare logic:

\medskip

\renewcommand{\labelenumi}{(\alph{enumi})}
\renewcommand{\theenumi}{(\alph{enumi})}

\textbf{Conjunction rule}

\begin{prooftree}
    \AxiomC{\SFHTriple{P_1}{C}{Q_1}}
    \AxiomC{\SFHTriple{P_2}{C}{Q_2}}
    \BinaryInfC{\SFHTriple{P_1 \land P_2}{C}{Q_1 \land Q_2}}
\end{prooftree}

\textbf{Disjunction rule}

\begin{prooftree}
    \AxiomC{\SFHTriple{P_1}{C}{Q_1}}
    \AxiomC{\SFHTriple{P_2}{C}{Q_2}}
    \BinaryInfC{\SFHTriple{P_1 \lor P_2}{C}{Q_1 \lor Q_2}}
\end{prooftree}

\medskip

Similar rules are known to be admissible for standard Hoare
logic~\cite{gordon:LectureSlides:HoareLogic}.
The soundness of both rules for access Hoare logic follows immediately.
Their admissibility is also true: If the premises for a rule are provable in
\ahl, then they are true by soundness.
But then the conclusion is also true by the soundness of the rule, which is therefore
provable in \ahl by completeness.

\section{Examples formalized in access Hoare logic}\label{sec:examplesFormal}

We revisit the examples from Section~\ref{sec:examples} in the context of our developed theory for access Hoare Logic.

\subsection{Access Security for Electronic Keys revisited}\label{sec:exampleHKformal}

\newcommand{\DK}{\deviceKey}
\newcommand{\CKo}{\cardKeyOne}
\newcommand{\CKt}{\cardKeyTwo}
\newcommand{\AC}{\access}
\newcommand{\code}[1]{\ensuremath{\mathtt{#1}}}

\newcommand{\Cod}{\mathrm{C1d}}
\newcommand{\Ctd}{\mathrm{C2d}}
\newcommand{\Acc}{\mathrm{Acc}}
\newcommand{\hPre}{\mathrm{Pre}}
\newcommand{\hPost}{\mathrm{Post}}

For our digital key system for hotel doors from Section~\ref{ex:HotelKeys}, we considered two readings of program~\verb|P0|, one being access secure (\verb|P1|) and the other not (\verb|P2|).
We now give a formal proof in \ahl{} to show that
\verb|P1| is access secure.  We also show that \verb|P2| is not access
secure by showing that the accress Hoare triple for \verb|P1| is not valid.

We first describe a derivation of 
\begin{equation}\label{atriple:hotel key P1}
 \SFHTriple{\hPre}{\ccifte{\code{not (dk==ck1)}}{\code{acc := (dk==ck2)}}{\code{(dk := ck2;acc := true)}}}{\hPost}
\end{equation}
in \ahl{}, showing that \verb|P1| is access secure wrt.~$\hPre$ and $\hPost$.
To make the formal proof more readable, we abbreviate some formulas:
\begin{align*}
    \Cod &:\quad \CKo = \DK  &
    \Ctd &:\quad \CKt = \DK\\
    \hPre &:\quad \Cod\lor\Ctd  &
    \hPost &:\quad \AC = \cctrue
\end{align*}

Using the assignment and consequence rules, we obtain
\begin{prooftree}
    \AxiomC{}
    \UnaryInfC{\SFHTriple{\hPost[\code{(dk==ck2)}/\AC]}{\code{acc := (dk==ck2)}}{\hPost}}
    \UnaryInfC{\SFHTriple{\Ctd}{\code{acc := (dk==ck2)}}{\hPost}}
\end{prooftree}
because $\hPost[\code{(dk==ck2)}/\AC] \equiv (\code{(dk==ck2)} = \cctrue)$ 
is semantically equivalent to $\Ctd$.

Similarly, using assignment, consequence and composition rules, we obtain 
\begin{prooftree}
    \AxiomC{}
    \UnaryInfC{\SFHTriple{\cctrue[\CKt/\DK]}{\code{dk := ck2}}{\cctrue}}
    \AxiomC{}
    \UnaryInfC{\SFHTriple{\hPost[\code{true}/\AC]}{\code{acc := true}}{\hPost}}
    \UnaryInfC{\SFHTriple{\cctrue}{\code{acc := true}}{\hPost}}
    \BinaryInfC{\SFHTriple{\cctrue}{\code{dk := ck2;acc := true}}{\hPost}}    
\end{prooftree}
because $\cctrue[\ldots] \equiv \cctrue$ and
$\hPost[\code{true}/\AC]\equiv (\cctrue = \cctrue)$ is equivalent to $\cctrue$.

Using consequence and conditional rules we finally obtain
\begin{prooftree}
    \AxiomC{\SFHTriple{\neg\Cod\limp\hPre}{\code{acc := (dk==ck2)}}{\hPost}}
    \AxiomC{\SFHTriple{\neg\neg\Cod\limp\hPre}{\code{dk := ck2;acc := true}}{\hPost}}
    \BinaryInfC{\SFHTriple{\hPre}{\ccifte{\code{not (dk==ck1)}}{\code{acc := (dk==ck2)}}{\code{(dk := ck2;acc := true)}}}{\hPost}}
\end{prooftree}
because $\Ctd$ semantically implies $\neg\Cod\limp\hPre$, 
and $\neg\neg\Cod\limp\hPre$ is equivalent to true.
This concludes the derivation of \eqref{atriple:hotel key P1} in \ahl.

\bigskip

We now show that \verb|P2| is not access secure wrt.~$\hPre$ and $\hPost$.
To this end we demonstrate that
\begin{equation}\label{atriple:hotel key P2}
 \SFHTriple{\hPre}{(\ccifte{\code{not (dk==ck1)}}{\code{acc := (dk==ck2)}}{\code{dk := ck2}})\code{;acc := true}}{\hPost}
\end{equation} 
is not valid by giving a counter model.
Using Soundness~\ref{thm:soundness} this also shows that \eqref{atriple:hotel key P2} is not provable in \ahl.

To define the counter model, let $a,b$ be different keys.  Define states $s, s'$ as follows:
\begin{align*}
    s&:\quad \CKo=\CKt=a,\ \DK=b,\ \AC=\ccfalse \\
    s'&:\quad \CKo=\CKt=a,\ \DK=b,\ \AC=\cctrue
\end{align*}
We compute that $\executes{\code{P2}}{s}{s'}$,
that $\hPre$ is false in $s$,
and that $\hPost$ is true in $s'$.
Hence $\SFHTriple{\hPre}{\code{P2}}{\hPost}$ is not valid.

\subsection{Access Security for Bitcoin revisited}\label{sec:exampleBitcoinformal}

\newcommand{\msg}{\codecommand{msg}}
\newcommand{\stack}{\codecommand{stack}}
\newcommand{\stackp}{\codecommand{stack'}}
\newcommand{\AgdaOpPush}{\mathtt{opPush}}
\newcommand{\trlp}[1]{\langle\!\langle #1 \rangle\!\rangle}
\newcommand{\trl}[1]{[\![ #1 ]\!]}
\newcommand{\opc}{\mathrm{opc}}
\newcommand{\sttop}{\code{top}}
\newcommand{\stpush}{\code{push}}
\newcommand{\stpop}{\code{pop}}
\newcommand{\stempty}{\code{empty}}

In this section we expand on the Bitcoin script example from
Section~\ref{ex:Bitcoin} by formally demonstrating that
$\sPK_\pKH$ is access secure.
To this end, we briefly introduce
the part of Bitcoin script relevant to  $\sPK_\pKH$,
a translation of that part to the \ccwhile language used in this paper,
and an explanation why the access Hoare tripe for the translated script is valid.

\subsubsection{Introduction to Bitcoin Script}%
\label{subsec:Bitcoin intro}

Bitcoin script is a stack-based scripting language~\cite{antonopoulos2023mastering},
given by a set of commands called \emph{opcodes}.
In this paper we focus on opcodes used in $\sPK_\pKH$,
a full list of opcodes can be found in~\cite{bitcoin:wiki}.

A number of opcodes are used to push data onto the stack;
we ignore details and write \code{<number>} for the sequence of
instructions that will push \code{number} onto the stack.
Values on the stack are also interpreted as truth values, 
in which case any value ${>}0$ will be interpreted as true, 
and any other value as false.

In addition to pushing data onto the stack, the following opcodes
are relevant to this paper: 
\begin{itemize}
\item \code{OP\_DUP} duplicates the top element of the stack.
\item \code{OP\_HASH} takes the top item of the stack and replaces it with its hash.
\item \code{OP\_EQUAL} removes the top two elements of the stack,
  checks whether they are equal, and pushes the result back onto the stack.
\item \code{OP\_VERIFY} invalidates the transaction if the top stack value is false. The top item on the stack will be removed.
\item \code{OP\_CHECKSIG} hashes the entire transaction, and checks whether the top two items on the stack form a correct pair of a signature and a public key for this hash.
\end{itemize}

Based on the introduces opcodes, the locking and unlocking
scripts of the pay-to-public-key-hash (P2PKH) script are given as follows:\\[.5ex]
\begin{tabular}{@{}ll}
\sPKpKH: & \verb|OP_DUP OP_HASH160 <pubKeyHash> OP_EQUAL OP_VERIFY OP_CHECKSIG| \\
\sSpKH:   & \verb|<sig> <pubKey>|
\end{tabular}\\[.5ex]
The unlocking script \sSpKH pushes
a signature \code{sig} and a public key \code{pubKey} onto the stack.
The locking script \sPKpKH checks whether \code{pubKey} 
hashes to the recorded \code{pubKeyHash}, and whether \code{sig} is a signature for the message signed by \code{pubKey}.

\subsubsection{Translation of Bitcoin script into a \ccwhile program}%
\label{subsec:Bitcoin Translation}

\begin{figure}[ht]
\begin{align*}
    \trlp{\code{<number>}} \quad=\quad & \code{stack\ :=\ push(number,stack);}\\
    \trlp{\code{OP\_DUP}} \quad=\quad & \code{fail\ :=\ \stempty(stack);\ \ x\ :=\ top(stack);}\\
     & \code{stack\ :=\ push(x,stack);}\\
    \trlp{\code{OP\_HASH160}} \quad=\quad & \code{fail\ :=\ \stempty(stack);\ \ x\ :=\ top(stack);}\\
     & \code{stack\ :=\ pop(stack);\ \ y\ :=\ hash160(x);}\\
     & \code{stack\ :=\ push(y,stack);}\\
    \trlp{\code{OP\_EQUAL}} \quad=\quad & \code{fail\ :=\ \stempty(stack);\ \ x\ :=\ top(stack);}\\
     & \code{stack\ :=\ pop(stack);\ \ fail\ :=\ \stempty(stack);}\\
     & \code{y\ :=\ top(stack);\ \ stack\ :=\ pop(stack);}\\
     & \code{b\ :=\ x == y;\ \ stack\ :=\ push(b,stack);}\\
    \trlp{\code{OP\_VERIFY}} \quad=\quad & \code{fail\ :=\ \stempty(stack);\ \ x\ :=\ top(stack);}\\
     & \code{stack\ :=\ pop(stack);}\\
     & \code{if\ x=0\ then\ fail\,{:=}\,true\ else\ skip;}\\
    \trlp{\code{OP\_CHECKSIG}} \quad=\quad 
     & \code{fail\ :=\ \stempty(stack);\ \ pbk\ :=\ top(stack);}\\
     & \code{stack\ :=\ pop(stack);\ \ fail\ :=\ \stempty(stack);}\\
     & \code{sig\ :=\ top(stack);\ \ stack\ :=\ pop(stack);}\\
     & \code{b\ :=\ correctsig(sig,msg,pbk);}\\
     & \code{stack\ :=\ push(b,stack);}
\end{align*}
\caption{The translation of opcodes to while programs.}\label{fig:transl}
\end{figure}

We first describe how to systematically translate Bitcoin script into the \ccwhile language.
We view scripts $P$ are lists of opcodes, and define the pre-translation $\trlp{P}$ by recursion: 
The empty script is the empty \ccwhile program, and $\trlp{\opc\ P} = \trlp{\opc}\code{;}\trlp{P}$
where $\trlp{\opc}$ is defined in Fig.~\ref{fig:transl}.  The full translation of a script $P$ then takes the possibility of $P$ to fail in the following way into account: Let \code{fail} be a fresh variable, and define
\begin{align*}
  \trl{P} \quad=\quad & \code{fail\ :=\ false;} \\
   & \trlp{P} \\
   & \code{if\ fail\, {==}\, true\ then\ push(0,stack)\ else\ skip}
\end{align*}

We now explain the translation of opcodes listed above and given in Fig.~\ref{fig:transl}.
We assume that the \ccwhile program has access to a variable \msg that stores the hash of the entire transaction.
Scripts are dealing with one data structure, the stack.  
In the \ccwhile language we represent this as a variable \code{stack} the stores a list, and the following three operations:
\code{\sttop(stack)},
\code{\stpush(v,stack)}, 
\code{\stpop(stack)}, which return: the first element of \stack; \stack with \code{v} added as the first element; and the tail of \stack, respectively.
Furthermore, we assume to have a number of operations that perform similar operations as their corresponding opcodes.  We also have \code{\stempty(stack)} that returns \code{true} if the stack is empty.

Overall, we obtain the \ccwhile program $P_{\sPK\pKH}=\trl{\sPK_\pKH}$, which takes as inputs message \msg and a list of integers \stack representing the stack, and outputs \stack as the updated stack.

\subsubsection{Access security of $P_{\sPK\pKH}$}%
\label{subsec:Bitcoin Access Security}

In order to express access security of $P_{\sPK\pKH}$, we start by formalizing the pre- and post-condition.
The post-condition expresses that the stack contains true as its top element:
\[  \hPost :\quad \neg\stempty(\stack) \wedge \sttop(\stack)>0 \]
while the pre-condition expresses that an address and a signature are
on the stack, where the address hashes to a value $pkh$ stored in the locking
script \sPKpKH, and the signature matches the address: 
\begin{align*}
    \hPre :\quad \exists\, pk,sig\, [\,&\stack = \stpush(pk,\stpush(sig,\stackp)) \\
    &  \land\ \code{hash160}(pk) = pkh  \\
    &  \land\ \code{correctsig}(sig,\msg,pk) >0 \,]
\end{align*}

By analyzing the translated script, we observe that 
\begin{equation}\label{atriple:bitcoin formal}
 \SFHTriple{\hPre}{P_{\sPK\pKH}}{\hPost}
\end{equation} 
is valid, that is, 
given any values for \msg, \stack, \stackp, such that $\hPost(\stackp)$ holds (i.e., the top element in \stackp is $>0$),
then we must have that $\hPre(\msg,\stack)$ holds.
Using Completeness~\ref{thm:completeness}, we therefore obtain that \eqref{atriple:bitcoin formal} is provable in \ahl.

\subsection{Access Security for the \texttt{CheckList} Program}\label{sec:exampleCLformal}

\newcommand{\clPre}{\mathrm{Pre}}
\newcommand{\clPost}{\mathrm{Post}}
\newcommand{\clInv}{\mathrm{Inv}}
\newcommand{\clInvp}{\mathrm{Inv'}}

We revisit the program in 
Section~\ref{ex:CheckList} and demonstrate that
\verb|CheckList| is access secure, by describing a derivation of the access Hoare triple for \verb|CheckList|.
We define
\begin{align*}
    S :\quad & \ccifte{\code{(el(i,L)==p)}}{\code{(acc := true)}}{\ccskip} \\
    P_1 :\quad & \ccwhiledo{\code{(i<=lh(L))}}{(S\code{;i := i+1})}\\
    \code{CheckList}:\quad & \code{i:=1;\ acc:=false;}\ P_1
\end{align*}
\verb|CheckList| uses inputs \ttp and \ttL to compute output \ttacc.
We also define
\begin{align*}
    \clPre :\quad & \exists j{\le}\lh(\ttL) ( \el(j,\ttL) = \ttp) \\
    \clPost :\quad & \ttacc=\tttrue \\
    \clInv :\quad & \ttacc=\tttrue\ \lor\
                        \exists j{\le}\lh(\ttL) ( j\ge i \land \el(j,\ttL) = \ttp) \\
    \clInvp :\quad & \ttacc=\tttrue\ \lor\
                        \exists j{\le}\lh(\ttL) ( j\ge i{+}1 \land \el(j,\ttL) = \ttp)
\end{align*}
Our task is to prove $\SFHTriple{\clPre}{\code{CheckList}}{\clPost}$ in \ahl.

Using the assignment rule we obtain
\begin{prooftree}
    \AxiomC{}
    \UnaryInfC{\SFHTriple{\clInv[\ttacc/\ttfalse, \tti/1]}{\code{i:=1}}{\clInv[\ttacc/\ttfalse]}}
    \AxiomC{}
    \UnaryInfC{\SFHTriple{\clInv[\ttacc/\ttfalse]}{\code{acc:=false}}{\clInv}}
    \BinaryInfC{\SFHTriple{\clInv[\ttacc/\ttfalse, \tti/1]}{\code{i:=1;acc:=false}}{\clInv}}
\end{prooftree}
As $\clInv[\ttacc/\ttfalse, \tti/1]\equiv \ttfalse=\tttrue\ \lor\
                        \exists j{\le}\lh(\ttL) ( j\ge 1 \land \el(j,\ttL) = \ttp)$ 
is semantically equivalent to $\clPre$, we can apply the Consequence Rule to obtain
\begin{equation}\label{atriple:CL1}
    \SFHTriple{\clPre}{\code{i:=1;acc:=false}}{\clInv}
\end{equation} 

The main part will be to derive
\begin{equation}\label{atriple:CL2}
    \SFHTriple{\clInv}{P_1}{\clPost}
\end{equation} 
in \ahl, then we can finish the overall derivation with an application of composition:
\begin{prooftree}
    \AxiomC{\eqref{atriple:CL1}}
    \noLine
    \UnaryInfC{\SFHTriple{\clPre}{\code{i:=1;acc:=false}}{\clInv}}
    \AxiomC{\eqref{atriple:CL2}}
    \noLine
    \UnaryInfC{\SFHTriple{\clInv}{P_1}{\clPost}}
    \BinaryInfC{\SFHTriple{\clPre}{\code{CheckList}}{\clPost}}
\end{prooftree}

For \eqref{atriple:CL2}, we first derive 
\begin{equation}\label{atriple:CL3}
    \SFHTriple{(i\le\lh(\ttL)\to\clInv)}{S}{\clInvp}
\end{equation}
We observe that $\clInvp[\ttacc/\tttrue]$ is valid,
and the same holds for
\[
  \el(i,\ttL) = \ttp \to (i\le\lh(\ttL) \to
    \clInv
  )
\]
Thus, we obtain using assignment and consequence
\begin{prooftree}
    \AxiomC{}
    \UnaryInfC{\SFHTriple{\clInvp[\ttacc/\tttrue]}{\code{acc:=true}}{\clInvp}}
    \UnaryInfC{\SFHTriple{\el(i,\ttL) = \ttp \to (i\le\lh(\ttL) \to \clInv)  }{\code{acc:=true}}{\clInvp}}
\end{prooftree}
Similar, as $\clInvp\to\clInv$ is valid, we obtain using the axiom for empty statement and consequence
\begin{prooftree}
    \AxiomC{}
    \UnaryInfC{\SFHTriple{\clInvp}{\ccskip}{\clInvp}}
    \UnaryInfC{\SFHTriple{\neg \el(i,\ttL) = \ttp \to (i\le\lh(\ttL) \to \clInv)  }{\ccskip}{\clInvp}}
\end{prooftree}
Hence, an application of the conditional rule yields \eqref{atriple:CL3}.  

Now we can continue as follows:
\begin{prooftree}
    \AxiomC{\eqref{atriple:CL3}}
    \noLine
    \UnaryInfC{\SFHTriple{(i\le\lh(\ttL)\to\clInv)}{S}{\clInvp}}
    \AxiomC{}
    \UnaryInfC{\SFHTriple{\clInvp}{\code{i:=i+1}}{\clInv}}
    \BinaryInfC{\SFHTriple{(i\le\lh(\ttL)\to\clInv)}{S;\code{i:=i+1}}{\clInv}}
    \UnaryInfC{\SFHTriple{\clInv}{P_1}{(\neg(i\le\lh(\ttL))\to\clInv)}}
\end{prooftree}
where the last rule is an application of while.
To conclude, we observe that $\clPost\to\clInv$ is valid, hence also
$\clPost\to(\neg(i\le\lh(\ttL))\to\clInv)$.
Therefore, an application of consequence yields \eqref{atriple:CL2}.

This finishes the formal derivation of
$\SFHTriple{\clPre}{\code{CheckList}}{\clPost}$
in \ahl.
\section{Comparison to other types of Hoare Logic}\label{sec:comparisons}

\newcommand{\post}{\mathrm{post}}
\newcommand{\pre}{\mathrm{pre}}
\newcommand{\SOP}{\mathrm{StrongestOverPost}}
\newcommand{\WUP}{\mathrm{WeakestUnderPost}}

We compare our approach to other types of reverse Hoare logic.
Another form of reversing Hoare logic has been proposed in the literature
\cite{vriesKoutavas:ReverseHoareLogic} and further developed to deal with
\emph{incorrectness} of programs \cite{OHearn:IncorrectnessLogic}.
\cite{OHearn:IncorrectnessLogic} defines under-approximate triples
$\IncTriple pcq$ that expresses that for each final state $s'\in q$
there is a state $s\in p$ such that program $c$, when started in state $s$, 
can terminate in state $s'$. 
This is a different perspective, which is incompatible to access Hoare triples.
In fact, as pointed out in \cite[Fact 9]{OHearn:IncorrectnessLogic}, under-approximations are not well-behaved in the context of backwards transformers.

\begin{figure}[ht]
\begin{minipage}{0.48\textwidth}
    \centering
\begin{tikzpicture}[>=latex,scale=.8]
\tikzstyle{pred}=[blue, font=\large\itshape]
\node[pred] (L) at (0,0)   {Predicates};
\node[pred] (R1) at (5,2)  {Predicates};
\node[pred] (R2) at (5,0)  {Predicates};
\node[pred] (R3) at (5,-2) {Predicates};
\draw[->, thick] (L) -- (R1)
  node[midway, above=1.5ex] {$\HoareTriple{-}c{-}$};
\draw[->, thick] (L) -- (R2) 
  node[midway, above] {$\post(c)$};
\draw[->, thick] (L) -- (R3)
  node[midway, below=1.5ex] {$\IncTriple{-}c{-}$};
\draw[->, thick] (R2) -- (R1)
  node[midway, right=0.7ex] {\huge \rotatebox[origin=c]{90}{$\subseteq$}};
\draw[->, thick] (R2) -- (R3)
  node[midway, right=0.7ex] {\huge \rotatebox[origin=c]{90}{$\subseteq$}};
\end{tikzpicture}
    \caption{Correctness and incorrectness triples in forward direction}\label{fig:forwards diag}
\end{minipage}
\hfill
\begin{minipage}{0.48\textwidth}
    \centering
\begin{tikzpicture}[>=latex,scale=.8]
\tikzstyle{pred}=[blue, font=\large\itshape]
\node[pred] (L) at (10,0)   {Predicates};
\node[pred] (R1) at (5,2)  {Predicates};
\node[pred] (R2) at (5,0)  {Predicates};
\node[pred] (R3) at (5,-2) {Predicates};
\draw[->, thick] (L) -- (R1)
  node[midway, above=1.5ex] {$\SFHTriple{-}c{-}^{-1}$};
\draw[->, thick] (L) -- (R2) 
  node[midway, above] {$\pre(c)$};
\draw[->, thick] (L) -- (R3)
  node[midway, below=1.5ex] {$\HoareTriple{-}c{-}^{-1}$};
\draw[->, thick] (R2) -- (R1)
  node[midway, right=0.7ex] {\huge \rotatebox[origin=c]{90}{$\subseteq$}};
\draw[->, thick] (R2) -- (R3)
  node[midway, right=0.7ex] {\huge \rotatebox[origin=c]{90}{$\subseteq$}};
\end{tikzpicture}
    \caption{Access security and correctness triples in backwards direction}\label{fig:backwards diag}
\end{minipage}
\end{figure}

To expand the comparison in more detail, we repeat the high-level overview given in \cite{OHearn:IncorrectnessLogic}, and contrast it with a comparable high-level overview for access Hoare triples.
Following \cite{OHearn:IncorrectnessLogic}, we consider
$\IncTriple -c-$ and $\HoareTriple -c-$ as relations on predicates,
and $\post(c)$ as a function on predicates, 
mapping each input predicate to the set of states reached upon termination.
Their relation is displayed in Fig.~\ref{fig:forwards diag}.
The picture represents a commuting diagram in the category of sets and binary relations, where 
$\IncTriple -c- = \post(c);\supseteq$ and 
$\HoareTriple -c- = \post(c);\subseteq$. 
Here, ‘;’ refers to the sequential composition of relations. 
The diagram  can be used to characterizes $\IncTriple -c-$ and $\HoareTriple -c-$
using $\post(c)$.
Also, a connection between $\post$, and strongest and weakest post-condition has been stated in \cite{OHearn:IncorrectnessLogic} in the following form: Defining
\begin{itemize}
    \item $\SOP(c)p \ =\  \bigwedge\{q|\HoareTriple pcq \text{ holds} \}$
    \item $\WUP(c)p \ =\  \bigvee\{q|\IncTriple pcq \text{ holds} \}$
\end{itemize}
we have
\[ \SOP(c) \quad=\quad \WUP(c) \quad=\quad \post(c)\]

Access Hoare triples offer a similar duality for the backwards perspective, see Fig.~\ref{fig:backwards diag}.
As already explained in \cite[Fact 9]{OHearn:IncorrectnessLogic}, this perspective is incomparable with under-approximation triples.
In Fig.~\ref{fig:backwards diag}, we consider
$\SFHTriple -c-^{-1}$ and $\HoareTriple -c-^{-1}$ as relations on predicates (from post- to pre-conditions),
and $\pre(c)$ as a function on predicates, 
mapping each input predicate $q$ to the set of states that can reach~$q$.
In the case of deterministic and terminating $c$, we have that 
the picture in Fig.~\ref{fig:backwards diag} represents again a commuting diagram, where 
$\SFHTriple -c-^{-1} = \pre(c);\subseteq$ and 
$\HoareTriple -c-^{-1} = \pre(c);\supseteq$. 
The diagram can again be used to characterizes $\SFHTriple -c-$ and $\HoareTriple -c-$
using $\pre(c)$.
Also, a connection between $\pre$, and strongest and weakest post-condition has been stated 
in Corollary~\ref{corollaryPwEqualPsIfTerminating}:
The weakest pre-conditions of $c$ for Hoare logic $P^w(c)$ coincides with the strongest pre-conditions of $c$ for access Hoare logic $P^s(c)$; 
moreover, they are equal to $\pre(c)$:
\[  P^w(c) \quad=\quad P^s(c) \quad=\quad \pre(c)  \]

Outcome logic \cite{Zilberstein:2025} provides a unified algebraic framework for reasoning about the diverse execution results of non-deterministic and buggy programs. It unifies Hoare and incorrectness logic by treating various termination modes---such as normal termination, assertion failure, and divergence---as distinct, composable outcomes. Like incompleteness logic, the focus is on the forward direction, and hence not directly comparable to access Hoare logic.

\section{Conclusion}

In this note, we made the point that
access security
cannot be addressed naturally by standard Hoare logic.
To overcome this issue, we introduced access Hoare logic
where the reasoning is in the reverse direction compared to Hoare logic,
namely from post-conditions to pre-conditions.
We gave three examples to demonstrate our point,
one regarding electronic keys,
a second considering the cryptocurrency Bitcoin,
and a third about granting access involving a key and a list of keys.
We introduced rules for access Hoare logic, and showed that these rules are
sound and complete.
We provided a link between access Hoare logic and Hoare logic by showing that weakest pre-conditions for Hoare logic coincides with strongest pre-conditions for access Hoare logic for total programs.
We also demonstrated a fundamental difference between access Hoare logic and incorrectness logic 
and related approaches, in that access Hoare logic provides a complementary view to standard Hoare logic in the context of backwards transformers, while incorrectness logic does this for forward transformers.

While access Hoare logic can be reduced to Hoare logic by introducing negation,
we argued that this would have several negative consequences,
including breaking verifiability using intuitionistic theorem provers ---
indeed, careful inspection shows that all proofs in this paper
(except for Theorem \ref{equivalence aHt and Ht},
which needs to be replaced by Remark \ref{equivalenceIntuitionistic})
only use intuitionistic logic,
which we have confirmed using Agda \cite{agda:Documentation},
an interactive theorem prover based on intuitionistic type theory.
Details of that formalization will appear in a separate publication.

There are, of course, many areas that we did not touch upon.
Basically, everything that has been investigated for Hoare logic should now
be revisited in the context of access Hoare logic.
Instead of aiming for a complete list, we mention two obvious areas.
In our approach to soundness and completeness, we take assertions to be semantical.
There is a vast literature on ways to make this more precise, by
fixing the syntax of the expressions in the language,
considering formal semantics, their axiomatization and implications on Hoare
logic~\cite{apt:1981}.
As a second area, the relationship to Dijkstra's predicate transformer
semantics and its variations need to be clarified~\cite{dijkstra:1975}.
Finally, detailed case studies are needed to demonstrate the
usefulness of access Hoare logic for
access security.

In connection to this paper, the authors have filed a patent entitled
\emph{Verifying Access Security of a Computer Program} \cite{beckmannSetzer:Patent:AccessHoareLogic:FilingUK}. This patent utilizes
verification condition generation for access Hoare logic, which will be
detailed in a separate publication.


\begin{acks}
We would like to thank Fahad F. Alhabardi.
This article builds upon and expands the use of weakest pre-conditions
for access
security,
which were explored during his PhD research
\cite{alhabardi:PhD:SmartContractsAgda}.
We would also like to thank John V Tucker for valuable comments on an earlier draft of this paper.
\end{acks}

\bibliographystyle{ACM-Reference-Format}
\bibliography{refsAccessHoare}

\end{document}